\documentclass{article}
\usepackage[utf8]{inputenc}
\usepackage{amsmath,amsfonts,braket,xcolor,graphicx}
\usepackage[breaklinks=true, colorlinks=true, linkcolor={blue!90!black}, citecolor={blue!90!black}, urlcolor={blue!90!black}]{hyperref}
\usepackage[margin=1in]{geometry}
\usepackage{multirow,doi}
\usepackage{amsthm}
\usepackage{listings}
\usepackage{bbold}
\usepackage{bm}
\usepackage{epsfig}
\newcommand{\vect}[1]{\boldsymbol{\mathbf{#1}}}
\newcommand{\nbhd}[1]{\mathcal{N}_{#1}}

\newcommand*{\Prob}{\mathsf{Pr}}
\newcommand{\pp}{\textup{p.p.}}

\newcommand{\wbar}{\overline{\left| w \right|}}
\newcommand{\Gate}[1]{\textsc{#1}}

\newcommand{\zgate}{\Gate{z}}

\newcommand{\xgate}{\Gate{x}}

\newcommand{\idgate}{\Gate{i}}

\usepackage{floatrow}
\usepackage{subfig}
\usepackage{layouts}

\AtBeginDocument{%
  }

\newtheorem{theorem}{Theorem}

\usepackage[symbol]{footmisc}

\usepackage[maxnames=5,style=numeric-comp,firstinits=true,sorting=none,backend=bibtex]{biblatex}
\renewbibmacro{in:}{}
\title{Parameter Transfer for Quantum Approximate Optimization\\ of Weighted MaxCut}
\author{%
    Ruslan Shaydulin$^{\dagger}$\footnotemark, Phillip C. Lotshaw$^{\ddag *}$, Jeffrey Larson$^{\P}$, James Ostrowski$^{\S}$, and Travis S. Humble$^{\ddag ||}$ %
    }
\date{$^{\dagger}${\small JPMorgan Chase, New York, NY, USA} \\
$^{\ddag}${\small Quantum Computational Science Group, Oak Ridge National Laboratory, Oak Ridge, TN 37830} \\
$^{\P}${\small Mathematics and Computer Science Division, Argonne National Laboratory, Lemont, IL 60439} \\
$^{\S}${\small Department of Industrial and Systems Engineering, University of Tennessee at Knoxville, Knoxville, TN 37996}\\
$^{||}${\small Quantum Science Center, Oak Ridge National Laboratory, Oak Ridge, TN 37830}\\
}

\begin{document}

\floatsetup[figure]{style=plain,subcapbesideposition=top}

\maketitle
\begin{abstract}
Finding high-quality parameters is a central obstacle to using the quantum approximate optimization algorithm (QAOA). Previous work partially addresses this issue for QAOA on unweighted MaxCut problems by leveraging similarities in the objective landscape among different problem instances. However, we show that the more general weighted MaxCut problem has significantly modified objective landscapes, with a proliferation of poor local optima. Our main contribution is a simple rescaling scheme that overcomes these deleterious effects of weights. We show that for a given QAOA depth, a single ``typical'' vector of QAOA parameters can be successfully transferred to weighted MaxCut instances. This transfer leads to a median decrease in the approximation ratio of only 2.0 percentage points relative to a considerably more expensive direct optimization on a dataset of 34,701 instances with up to 20 nodes and multiple weight distributions. This decrease can be reduced to 1.2 percentage points at the cost of only 10 additional QAOA circuit evaluations with parameters sampled from a pretrained metadistribution, or the transferred parameters can be used as a starting point for a single local optimization run to obtain approximation ratios equivalent to those achieved by exhaustive optimization in $96.35\%$ of our cases.
\end{abstract}

\footnotetext{$^{*}$Both authors contributed equally to this work and are ordered randomly.}

\section{Introduction}

With the capabilities of quantum computers rapidly improving~\cite{arute2019quantum,Google2021QAOA}, researchers widely believe that quantum algorithms may soon become competitive with classical state-of-the-art methods on scientifically relevant problems~\cite{Alexeev2021}. A promising application domain for near-term quantum computers is combinatorial optimization, with the quantum approximate optimization algorithm (QAOA)~\cite{Hogg2000,farhi2014quantum} as a leading candidate algorithm. QAOA is a hybrid quantum-classical algorithm that combines a parameterized quantum evolution with a classical method for obtaining high-quality parameters so that the measurement outcomes of the parameterized quantum state correspond to good solutions for the target combinatorial optimization problem.

The performance of QAOA depends crucially on the choice of these parameters. The problem of identifying good QAOA parameters has attracted considerable interest, with large bodies of theoretical \cite{wurtz2021cd,2110.10685,Farhi2020SK,Wurtz2021Bounds}
and computational  \cite{Lykov2020tensorqaoa,Medvidovic2021QAOA54qubit,Shaydulin2020Symmetries,Shaydulin2020CaseStudy,zhou2020quantum,crooks2018performance,shaydulin2019multistart,Shaydulin2021Exploiting}
work. Several recent results show that as a function of the parameters, the QAOA energy landscape is nearly instance independent for typical problem instances that come from a reasonable distribution. This makes it possible to \emph{transfer} parameters from a set of preoptimized problem instances to new instances, thereby removing or reducing the cost of parameter optimization~\cite{Lotshaw2021BFGS,Galda2021transfer,wurtz2021fixedangle,brandao2018concentration}. However, these prior results address only unweighted MaxCut problems and the Sherrington--Kirkpatrick (SK) model, which are of limited practical interest. 

In this work we demonstrate parameter transfer for QAOA applied to the weighted MaxCut problem, which has many relevant practical problems as special cases~\cite{Lucas2014qubo}, including network community detection~\cite{Newman06modularity}. 
We develop intuition by first examining the impact of weights on the QAOA parameter landscapes and then propose a simple transfer scheme based on this intuition.  Our scheme, illustrated in Fig.~\ref{fig:schematic}, uses a simple rescaling technique to transfer typical (median) optimized parameters from unweighted MaxCut instances to weighted instances with varying sizes and with weights coming from varying distributions. Despite its simplicity, our approach leads to approximation ratios that are competitive with those found by  direct (computationally expensive) parameter optimization. We show how the performance can be further improved at the cost of only a few extra QAOA circuit evaluations by sampling from a metadistribution of good parameters learned using kernel density estimation. If the transferred parameters are used as the starting point for the BFGS optimization algorithm, a single BFGS run recovers approximation ratios equal to (up to high precision) that of QAOA with directly optimized parameters in $96.35\%$ of the cases considered. Our results provide empirical evidence that instance independence of optimal parameters, previously shown for unweighted MaxCut, extends to the weighted case. Our numerical results are for QAOA with 1, 2 and 3 layers, although it is likely that the same qualitative behavior will be observed for a higher number of layers. To show the limitations of parameter transfer, we explore a set of graphs with edge weights drawn from a pathological distribution and identify the mechanisms behind the relatively poor performance of the transferred parameters for these instances. While these cases challenge our direct parameter transfer approach, we find that our parameters can still serve as a good initial point for BFGS optimization even in such instances.

All data in this paper is available in the open-source software package \texttt{QAOAKit}~\cite{shaydulin2021qaoakit}, and we provide code examples of how our findings can be leveraged in practice~\cite{codeexampledirectlink}. Our findings demonstrate that a simple transfer scheme can dramatically reduce the cost of parameter optimization for QAOA in the general weighted case, bringing QAOA one step closer to being able to address problems of practical interest and demonstrate quantum computational advantage.

 \begin{figure}[t]
    \centering
    \includegraphics[width=\columnwidth,trim={0 0cm 0cm 0cm},clip]{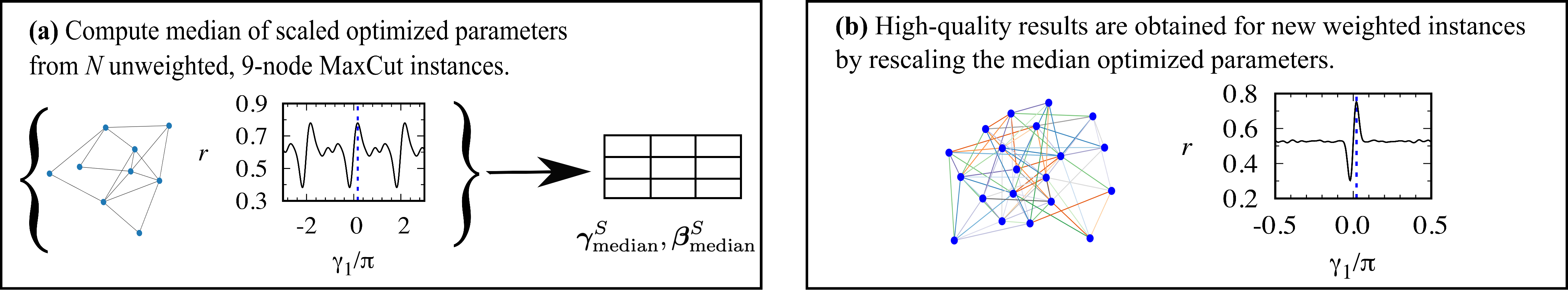}
    \caption{Example of parameter transfer for QAOA on weighted MaxCut. (a) We use a previous set of $N = 261,080$ optimized, unweighted, 9-node MaxCut instances to compute median (``typical'') scaled parameters $(\bm \beta_\mathrm{median}^S,\bm\gamma_\mathrm{median}^S)$ for each QAOA depth $p$ (Table \ref{table:median} and Eq.~\ref{eq:scale_median}). We develop an approach in Eqs.~\ref{eq:scale_weighted_beta}, \ref{eq:scale_weighted_gamma} to rescale the $\bm \beta_\mathrm{median}^S$ and $\gamma_\mathrm{median}^S$ to obtain high-quality parameters for new weighted instances such as the 20-node weighted instance in (b). The blue dashed line in (a) indicates directly optimized $\gamma_1$ and in (b) the transferred $\gamma_1$.} 
    \label{fig:schematic}
\end{figure}

\section{Methods}
\subsection{Setup}

We begin by briefly defining the relevant concepts and establishing notation. We consider optimization problems where the goal is to find a binary string $z\in \{0,1\}^n$ that maximizes some objective function $\mathcal{C}(z)$.
 The objective function 
 can be encoded as an $n$-qubit Hamiltonian $C$ that acts on computational basis states as $C \vert z \rangle = \mathcal{C}(z)\vert z \rangle.$
 For a given problem instance, we denote the maximum value of the objective function $\mathcal{C}_\mathrm{max}=\max_z \mathcal{C}(z)$ and the minimum $\mathcal{C}_\mathrm{min}=\min_z\mathcal{C}(z)$. In numerical experiments below, the values $\mathcal{C}_\mathrm{max}$ and $\mathcal{C}_\mathrm{min}$ are obtained by complete enumeration.

The quantum approximation optimization algorithm  
approximately solves such combinatorial optimization problems by preparing a parameterized quantum state
\begin{equation} \label{eq:QAOA}
\vert \bm \beta, \bm \gamma, \mathcal{C}\rangle = \prod_{l=1}^p e^{-i \beta_l B} e^{-i \gamma_l C} \vert + \rangle^{\otimes n}, 
\end{equation}
where $B = \sum_{j=1}^n\xgate_j$ 
is a sum of single-qubit Pauli $\xgate$ operators. The number $p$ is commonly referred to as the QAOA depth. The parameters $\bm \beta = (\beta_1,...,\beta_p)$ and $\bm \gamma = (\gamma_1,...,\gamma_p)$ are chosen so that the (approximate) solution to the optimization problem is retrieved upon measuring $\vert \bm \beta, \bm \gamma\rangle$. This is typically done by choosing $(\bm \beta, \bm \gamma)$ that maximize the expected objective value of the measurement outcomes:
\begin{equation} \label{eq:objective_expectation}
\langle C (\bm \beta, \bm \gamma) \rangle =  \bra{\bm \beta, \bm \gamma, \mathcal{C}}C\ket{\bm \beta, \bm \gamma, \mathcal{C}} = \sum_{z\in \{0,1\}^n}\Prob(z)\mathcal{C}(z). 
\end{equation}
Below we follow Ref.~\cite{brandao2018concentration} and refer to the value in \eqref{eq:objective_expectation} as the ``QAOA objective'' or simply ``objective'' for brevity. For a fixed function $\mathcal{C}$ and depth $p$, we can vary $\bm \beta$ and $\bm \gamma$ in \eqref{eq:objective_expectation} to produce an objective landscape~\cite{brandao2018concentration}. We quantify the performance on a given instance for a given $(\bm \beta, \bm \gamma)$ using the approximation ratio
\begin{equation} \label{eq:approx_ratio}
r = \frac{\langle C(\bm \beta, \bm \gamma) \rangle - \mathcal{C}_\mathrm{min}}{\mathcal{C}_\mathrm{max}-\mathcal{C}_\mathrm{min}} 
\end{equation}
with $0 \leq r \leq 1$. When discussing the difference between two approximation ratios we sometimes present the difference in terms of percentage points for clarity of presentation. We say that the gap between $r_1 > r_2$ is $k$ percentage points ({\pp}) if $(r_1-r_2)\times 100 = k$.

In this paper we consider the weighted MaxCut problem, which has many  interesting practical problems as special cases \cite{Newman06modularity,Lucas2014qubo}. In weighted MaxCut, given an undirected graph $G=(V,E)$, with each edge $(i,j) \in E$ having a weight $w_{i,j}=w_{j,i}$, the goal is to partition the set of vertices $V$ into two parts to maximize the sum of weights of edges that span both parts. The objective function of MaxCut is represented by the Hamiltonian
\begin{equation} 
C_{\text{MaxCut}} = \frac{1}{2}\sum_{( i,j )\in E} w_{i,j}\left(\idgate-\zgate_i\zgate_j\right), 
\end{equation} 
where in the unweighted case, $w_{i,j} = 1, \forall (i,j)\in E$. For a given graph, we denote the average absolute edge weight by
\begin{equation}
    \wbar = \frac{1}{|E|}\sum_{( i,j )\in E} \left|w_{i,j}\right|.
\end{equation}

\subsection{Objective landscapes with weights}

We now develop the intuition for how the introduction of weights affects QAOA objective landscapes as a function of the QAOA parameters $(\bm\beta,\bm\gamma)$. We begin  by considering an example of a triangle-free graph with depth $p=1$, where an analytical formula for $\langle C (\bm \beta, \bm \gamma) \rangle$ is available. We then extend to more general instances and show how it motivates our parameter transfer scheme. We evaluate how the proposed parameter transfer scheme generalizes to the dataset of 34,701 graph instances numerically in Sec.~\ref{sec:numerical_experiments}.

\subsubsection{Triangle-free graph}

We begin by considering the simple case of a triangle-free weighted graph with $p=1$, for which the QAOA objective can be expressed as
\begin{equation} \label{eq:qaoa_exp_triangle_free}
  \small
  \begin{aligned}
    \langle C(\beta_1,\gamma_1)\rangle & = \frac{W}{2} + \frac{\sin 4 \beta_1}{4} \sum_{( i,j ) \in E} w_{i,j} \sin (w_{i,j} \gamma_1) 
    \left(\prod_{l \in \nbhd{i} \backslash\{j\}} \cos ( w_{i,l}\gamma_1)+\prod_{k \in \nbhd{j} \backslash\{i\}} \cos (w_{j, k}\gamma_1 )\right), 
  \end{aligned}
\end{equation}
where $W=\sum_{( i,j )\in E}w_{i,j}$ and $\nbhd{i} \backslash\{j\}$ is the neighborhood of vertex $i$ excluding vertex $j$ \cite{hadfield2018quantum}. With the fixed optimal $\beta_1=\frac{\pi}{8}$, the objective is composed of terms that oscillate in $\gamma_1$ at frequencies set by the edge weights and may not be periodic for irrational $w_{i,j}$. Figure \ref{fig:triangle free} shows an example of the behavior for a 4-node cycle graph, with (a) the individual components $w_{i,j} \sin (w_{i,j} \gamma_1)\cos (w_{j, k}\gamma_1)\sin(4\beta_1)/4$ and (b) the approximation ratio both given as a function of $\gamma_1$.

\begin{figure}
    \centering
    \includegraphics[width=7cm, height=5cm, keepaspectratio,trim={0.5cm 2cm 0.5cm 1.6cm}, clip]{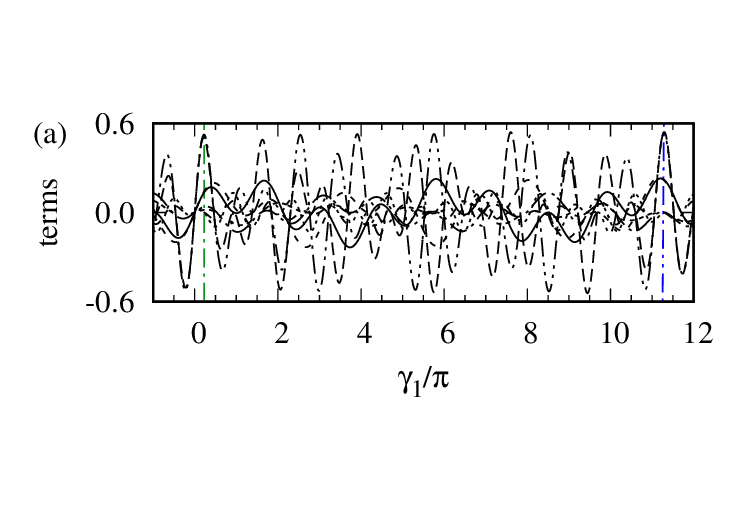}
    \includegraphics[width=7cm, height=5cm, keepaspectratio,trim={0.65cm 2cm 0.5cm 1.6cm}, clip]{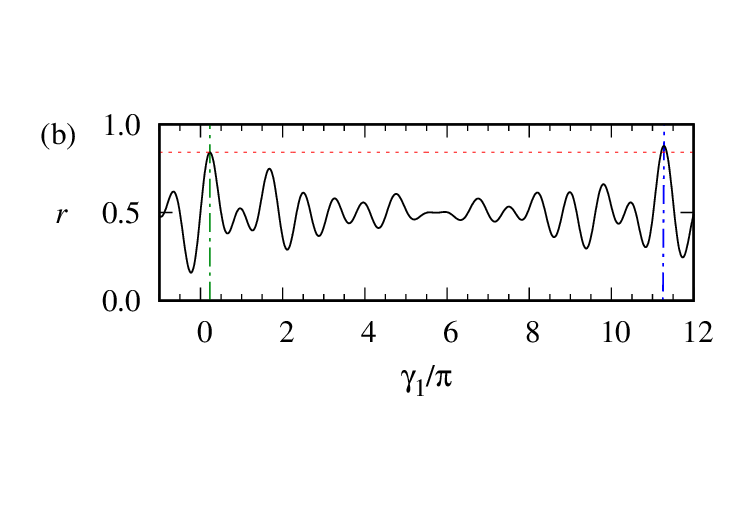}
    \caption{Objective for a 4-node cycle graph with edge weights $w_{0,1}=0.94, w_{1,2}=-0.53, w_{2,3}=-2.17, w_{3,0}=0.36$ and $\beta_1=\pi/8$. (a) The eight individual terms $w_{i,j}\sin(w_{i,j}\gamma_1)\cos(w_{i,k}\gamma_1)/4$ in Eq.~(\ref{eq:qaoa_exp_triangle_free}) and (b) the approximation ratio. Vertical dashed-dotted lines show the locations of two local optima, and the horizontal dotted line shows the optimum approximation ratio near $\gamma_1=0.$}
    \label{fig:triangle free}
\end{figure}

The green dashed-dotted-line in Fig.~\ref{fig:triangle free}(b) shows that a high approximation ratio is obtained at $\gamma_1\approx 0.23\pi$, because of an approximate alignment and maximization of the terms in Fig.~\ref{fig:triangle free}(a). The appearance of such a maximum at small values of $\gamma_1$ is a generic feature of \eqref{eq:qaoa_exp_triangle_free}. It arises from an approximate alignment in the phases of the various terms, where small positive values of $\gamma_1$ produce positive values of both $w_{i,j} \sin(w_{i,j} \gamma_1)$ and $\cos(w_{i,j} \gamma_1)$, independent of the sign or magnitude of $w_{i,j}$. By contrast, larger choices of $\gamma_1$ typically give much smaller $r$, since the phases can have varying or canceling oscillation frequencies that depend on the set of edge weights. Overall, this phase alignment for small values of $\gamma_1$ results in small values of $\gamma_1$ typically being a good parameter choice.
 
Of course, small positive values of $\gamma_1$ do not necessarily maximize \eqref{eq:qaoa_exp_triangle_free} globally. We have constructed such an example in Fig.~\ref{fig:triangle free}(b), where the blue dashed-double-dotted line at $\gamma_1 \approx 11.25\pi$ gives a superior approximation ratio to the best optimum for small values of $\gamma_1$ at $\gamma_1 \approx 0.23\pi$. The peak at $\gamma\approx11.25\pi$ is obtained by simultaneously approximately maximizing terms with the two prefactors $w_{i,j}$ of largest magnitude in Fig.~\ref{fig:triangle free}(a), which gives a higher objective value than does the approximate alignment of terms at a small value of $\gamma_1$.  Similar maxima at large $\gamma_1$ can be devised for other triangle-free graphs with $p=1$ using \eqref{eq:qaoa_exp_triangle_free}. In general, however, finding such maxima requires optimizing over an unbounded domain. Therefore, in numerical experiments we restrict the starting points for the local optimizations to $-\pi/\wbar \leq \gamma_l \leq \pi/\wbar$.

\subsubsection{Generic graph}

We now extend the motivation from the triangle-free case to generic graphs. In the triangle-free $p=1$ case, we found a maximum in $r$ at small values of $\gamma_1$ that we attributed to the approximate phase alignment in the component terms of the QAOA objective in \eqref{eq:qaoa_exp_triangle_free}. In more general cases, the objective is instead given by \eqref{eq:objective_expectation}.  This equation depends on the operators in \eqref{eq:QAOA}, which at the $l$th step can be expressed in terms of trigonometric functions of the parameters and weights:
\begin{equation}
    e^{-i \beta_l B} = \prod_{j=1}^n \big[\cos{\beta_l} - i\xgate_j\sin{\beta_l} \big],
    \hspace{0.1in}
    \mbox{ and}
    \hspace{0.1in}
     e^{-i \gamma_l C}  = \prod_{( i,j)\in E} \big[\cos{\frac{\gamma_l w_{i,j}}{2}} + i\zgate_i\zgate_j\sin{\frac{\gamma_l w_{i,j}}{2}} \big].
\end{equation}
As in \eqref{eq:qaoa_exp_triangle_free}, the objective value is again determined by terms that oscillate in $w_{i,j}\gamma_l$. Intuitively, we can expect the phases in these terms to approximately align for small $\gamma_l$, while larger $\gamma_l$ may effectively randomize the phases and the overall behavior. Hence, we expect a high-quality maximum from the alignment of terms near $\gamma_l=0$. We observe this maximum numerically for all $l=1,\ldots,p$.

Figure~\ref{fig:contours small} demonstrates how these effects impact the QAOA objective landscapes for a single, randomly selected, 14-node (not triangle-free) graph with edge weights drawn from different distributions. We plot the approximation ratio $r$ for $p=1$ QAOA as a function of $(\beta_1,\gamma_1)$. Figure~\ref{fig:contours small}(a) considers an unweighted graph, which has a simple periodic structure, while (b) shows the same graph but with weights drawn from an exponential distribution with a mean weight of 5.  Here we see the appearance of a sharp local maximum at small values of $\gamma_1$ while larger $\gamma_1$ give suboptimal $r$ with small fluctuations. Thus the high-quality maximum at small $\gamma_1$ from approximate phase alignment is also manifest in this more general example.

The maximum in Fig.~\ref{fig:contours small}(b) is surrounded by a small basin of attraction, making it difficult to find optimal parameters when starting from random initial parameters. Furthermore, the optimized $\gamma_1$ in (b) is much smaller than in (a); hence, naively using optimized parameters from (a) in (b) will result in a poor approximation ratio.  To allow the transfer from unweighted to weighted graphs, we must take into account the magnitude of edge weights. Specifically, we observe that an optimal %
value of $\bm \gamma$ changes with the mean absolute value of the edge weights as $\bm\gamma \approx O(1 / \wbar)$, irrespective of the depth $p$. 

To see this, first consider QAOA applied to two general (not necessarily MaxCut) objective functions $\mathcal{C}_w$ and $\mathcal{C}$ that are equivalent up to some factor $w > 0$, in other words, $\mathcal{C}_w = w\mathcal{C}$. If QAOA achieves an approximation ratio of $r$ for $\mathcal{C}$ with parameters $(\bm \beta, \bm \gamma)$, then QAOA will achieve an identical approximation ratio for $\mathcal{C}_w$ with $(\bm \beta, \bm \gamma/w)$. We give a rigorous proof of this statement in the Appendix \ref{methods}. Note that this statement holds for arbitrary depth $p$. 

This observation informs our intuition about typical problems and cases. If the objective function is on average increased by a factor $w > 0$, then $\bm \gamma^*$ should decrease by the same factor. Equivalently, the objective function can be rescaled and $\bm\gamma^*$ fixed. In the case of MaxCut, the objective function scales linearly with graph edge weights; therefore, the relevant metric is the average absolute edge weight $\wbar$.

\begin{figure}
    \centering
    \includegraphics[width=\columnwidth,trim={0cm 0cm 0cm 0cm},clip]{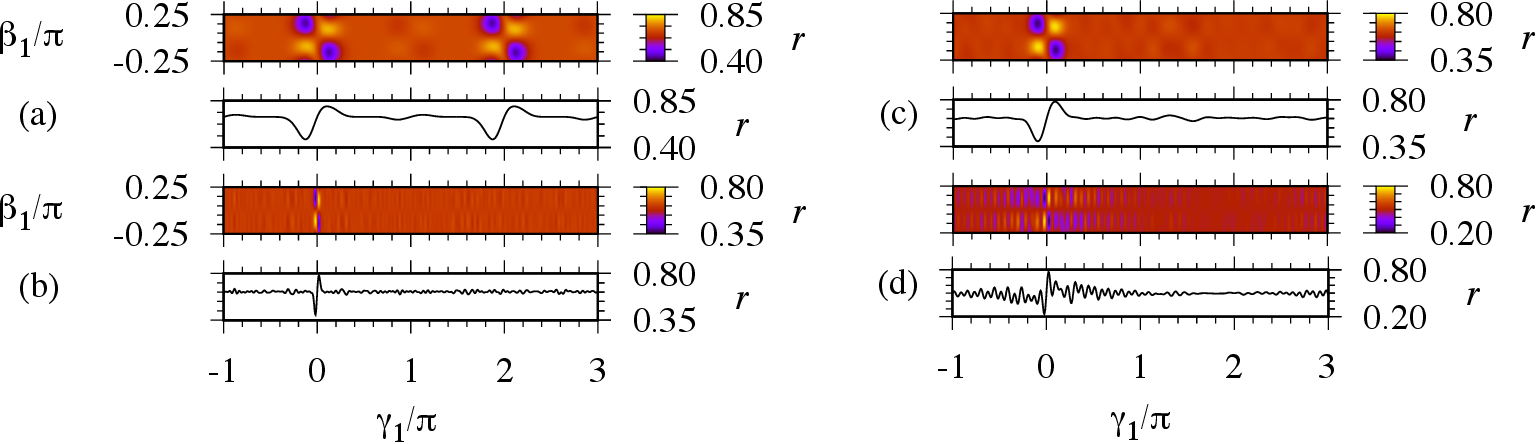}
    \caption{Approximation ratio landscapes $\langle C(\beta_1,\gamma_1)\rangle$ at $p=1$ and scans $\langle C(\beta_1^*,\gamma_1)\rangle$ at optimized $\beta_1^*$ for (a) unweighted MaxCut on a 14-vertex graph, (b) weighted MaxCut for the same graph with edge weights drawn from an exponential distribution with mean absolute weight $\wbar=5$, (c) exponential (same as (b)) rescaled to have $\wbar=1$,  and (d) Cauchy distributions. The $\beta_1$ ranges are the same in all panels.} 
    \label{fig:contours small}
\end{figure}

In Fig.~\ref{fig:contours small}(c) we demonstrate how weight magnitudes affect the objective landscape. The figure 
shows the same weight instance as (b) but with a uniform rescaling in $\mathcal{C}$ such that $\wbar=1$ (the same plot could also be generated by fixing $\mathcal{C}$ and rescaling $\gamma_1$). This ``stretches'' the landscape, alleviating the issue of a narrow basin of attraction in (b).  The objective has a similar structure to the unweighted case near $\gamma_1=0$, which has the same average edge weight. The optimized parameters from the unweighted case can now be transferred to the weighted case. We find similar behavior to (b),(c) both for weights drawn from a uniform distribution of positive values and for weights drawn from a uniform distribution of positive and negative values. The optimized $\gamma_1$ from the unweighted case can be successfully transferred to each of these varying weighted cases. The scaling factor $\wbar$ is essential in achieving the correspondence between landscapes that enables parameter transfer. 

The final plot (d) shows an instance with weights drawn from a Cauchy distribution. The Cauchy distribution on $[-\infty,\infty]$ is well known as an exemplar of pathological behavior, with symmetry about the origin and normalization but with an undefined mean and higher moments. In this work we use a truncated Cauchy distribution, with weights $-10^3 \leq w_{i,j} \leq 10^3$ before rescaling to $\wbar=1$, which avoids the pathologies of the moments but achieves a much higher variance for the edge weights than other distributions considered so far. The landscape is visibly ``rugged'' with many local maxima, although we still observe a prominent maximum near $\gamma_1 = 0$.  Below, we show evidence that even for these instances we can often determine high-quality parameters by transferring the optimized parameters from the unweighted case and using them as a starting point for a local optimizer. 

To summarize, we identified two important aspects of the objective landscapes for weighted graphs.  The first is the appearance of a prominent local maximum at small $\gamma_1$ values, related to phase alignment in components of the objective.  The second is the overall magnitude of the objective function, which can be standardized between instances by rescaling all edge weights by a factor $\wbar$. When the variance of the edge weights is not too large, this rescaling aligns the objective maximum of weighted instances at small $\bm\gamma$  with the corresponding maximum of unweighted instances, allowing parameter transfer.  The examples we give are for $p=1$. We expect similar arguments and intuition to apply to $p>1$, and we present numerical evidence for $p\in\{1,2,3\}$ in Sec.~\ref{sec:numerical_experiments}. 

\subsection{Scaling rule for parameter transfer}
 
We now present our transfer procedure. An overview is shown in Fig.~\ref{fig:schematic}. As the intuition developed above suggests, to transfer QAOA parameters from unweighted to weighted MaxCut instances, we must rescale the parameters to account for the change in the objective function values. We account for this change by incorporating the average edge weight (following the intuition outlined above) and the average degree of the graph. The parameter scaling with the average degree improves performance when transferring parameters between different graphs; it has been observed previously for the unweighted MaxCut and the SK model~\cite{Ozaeta2021IsingQAOA,Farhi2020SK,Wang2018} and proven rigorously in the infinite-size limit~\cite{2110.10685,2110.14206}.

Consider a dataset of preoptimized parameters $\{\bm\beta^{*}_j, \bm\gamma^{*}_j\}_{j=1}^N$ for $N$ unweighted graphs. We begin by rescaling the parameters in the dataset using the average degree $d_j = \frac{\sum_{k\in V}d_j^k}{|V|}$ of the graph: 
\begin{equation}
    (\bm \beta^{S}_j, \bm\gamma^{S}_j) = \left(\bm\beta^{*}_j, \frac{\bm\gamma^{*}_j}{\arctan(\frac{1}{\sqrt{d_j-1}})}\right) \mbox{for } j = 1, \ldots, N.
    \label{eq:scale_median}
\end{equation}
We scale $\bm\gamma$ by $\arctan(\frac{1}{\sqrt{d-1}})$ following the exact formula for optimal triangle-free and $p=1$ $\bm\gamma$ from \cite{Wang2018} because we observe it to empirically give a better performance for the low values of $p$ considered in this work. The use of the $\arctan$ scaling here is a heuristic choice, as we consider non-triangle-free and $p>1$ instances for transfer, which do not conform to the analytic assumptions of \cite{Wang2018}. Nonetheless we find this same scaling to be useful in achieving high-quality results for our cases. For higher values of $p$, scaling by $\sqrt{d}$ may be preferable~\cite{Ozaeta2021IsingQAOA,2110.10685,2110.14206}, although the two are asymptotically equivalent for large $d$. Next, the median scaled parameters $(\bm \beta^S_\mathrm{median}, \bm \gamma^S_\mathrm{median})$ are computed from $(\bm\beta^{S}_j, \bm\gamma^{S}_j)_{j=1}^N$. This completes steps pictured in Fig.~\ref{fig:schematic}(a).

Finally, to determine transferred angles for a new instance, the pre-computed median parameters are rescaled as
\begin{align}
 \bm\beta_{w} & = \bm \beta_{\text{median}}^S,
 \label{eq:scale_weighted_beta}\\
 \bm\gamma_{w} & =  \bm\gamma_{\text{median}}^S\frac{\arctan(\frac{1}{\sqrt{d_w-1}})}{\wbar}  \label{eq:scale_weighted_gamma}. 
\end{align}
Here $d_w$ is the average degree and $\wbar$ is the average absolute edge weight of the weighted graph to which the parameters are transferred. Equations~(\ref{eq:scale_weighted_beta})--(\ref{eq:scale_weighted_gamma}) determine the transferred parameters as in Fig.~\ref{fig:schematic}(b).

In this work we use the dataset containing preoptimized parameters for all $N=261,080$ nonisomorphic connected 9-node  graphs~\cite{Lotshaw2021BFGS,shaydulin2021qaoakit}. The median scaled parameters from this dataset $\bm\beta_{\text{median}}^S$ and $\bm\gamma_{\text{median}}^S$ are given in Table~\ref{table:median}. With the pre-computed median parameters it is straightforward to generate transferred angles for new weighted instances using (\ref{eq:scale_weighted_beta})-(\ref{eq:scale_weighted_gamma}). We evaluate the success of this simple procedure for a variety of new instances in the next section.

\begin{table}[]
\renewcommand{\arraystretch}{1.5}
\begin{tabular}{|c|cccc|}
\hline
  p & &1 &2 &3 \\ \hline
\multirow{2}{*}{\;$1$\;}  & $\bm \beta^S_\mathrm{median}$ \hspace{-0.2in}  & \; -0.101708$\pi$\; & & \\ 
                          &$\bm \gamma^S_\mathrm{median}$ \hspace{-0.2in} & \; -0.287231$\pi$\; & & \\ \hline
\multirow{2}{*}{\;$2$\;} & $\bm \beta^S_\mathrm{median}$ \hspace{-0.2in}  & \; -0.139136$\pi$\; & \; -0.083772$\pi$\; &  \\  
                       &$\bm \gamma^S_\mathrm{median}$ \hspace{-0.2in}  & \; -0.230102$\pi$\; & \; -0.453701$\pi$\; & \\ \hline
\multirow{2}{*}{\;$3$\;} & $\bm \beta^S_\mathrm{median}$ \hspace{-0.2in}   & \; -0.149780$\pi$\; & \; -0.107380$\pi$\; & \; -0.063381$\pi$\;  \\    
                        &$\bm \gamma^S_\mathrm{median}$ \hspace{-0.2in} & \; -0.199343$\pi$\; & \; -0.389866$\pi$\; & \; -0.466856$\pi$\; \\ \hline
\end{tabular}
        \caption{Median $\bm \beta^S_\mathrm{median}, \bm \gamma^S_\mathrm{median}$ QAOA parameters computed over optimized parameters for all nonisomorphic 9-node graphs \cite{shaydulin2021qaoakit,Lotshaw2021BFGS}, with $\bm\gamma$ scaled following Eq.~\ref{eq:scale_median}. 
        }
        \label{table:median}
\end{table}

\begin{figure*}[t]
    \centering
    \includegraphics[width=0.9\textwidth]{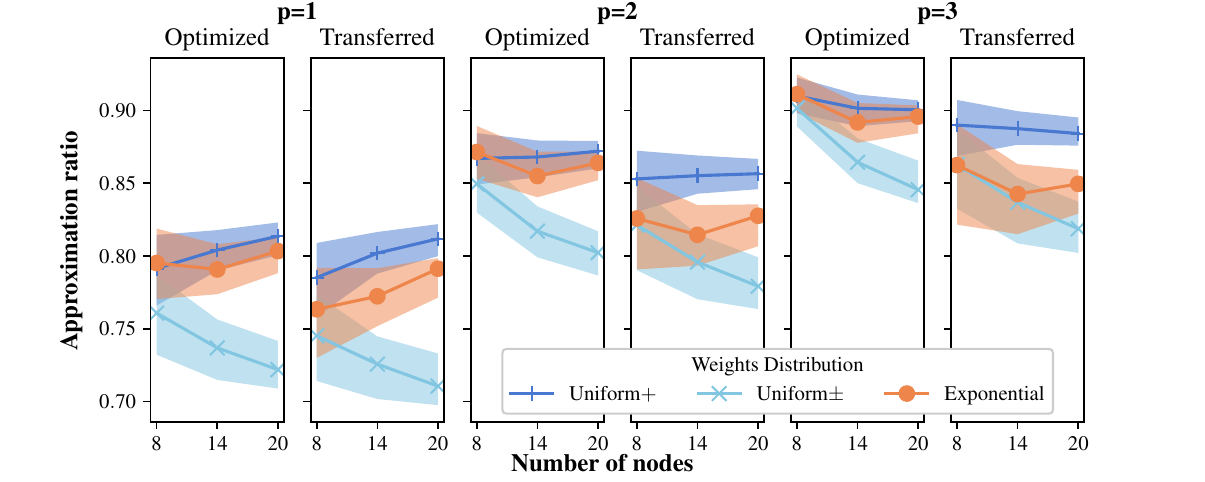}
    \caption{Approximation ratios using directly optimized parameters and transferred parameters, for $p=1,2,3$ and varying weight distributions. Markers indicate median values, shaded regions are interquartile range, and lines guide the eye. The median (over the full dataset) decrease in the approximation ratio is only 2.0 {\pp}~with the transferred parameters. Using KDE reduces this gap to 1.2 {\pp}, and BFGS reduces it further, as described in the text.}
    \label{fig:approximation_ratios_transfer}
\end{figure*}

\section{Results of Numerical experiments}\label{sec:numerical_experiments}

We evaluate the proposed parameter transfer scheme on a dataset of $34,701$ weighted graphs with up to 20 nodes. The dataset includes all nonisomorphic graphs on 8 nodes, 300 Erd\H{o}s-R\'enyi random graphs on 14 nodes, and 50 Erd\H{o}s-R\'enyi with 20 nodes, where these random graphs have an edge-creation probability of $0.5$. The edge weights are drawn from three distributions: uniform on $(0,b)$, uniform on $(-b,b)$, and exponential. 
Without loss of generality, we set $b=1$ to generate weights and  then rescale the edge weights in each instance so that $\wbar=1$.  To provide a baseline against which to compare our transfer proposal, we obtain optimized parameters and QAOA approximation ratios by exhaustive optimization for each weighted graph. The precise definition of the weight distributions and an overview of the parameter optimization are given in Appendix \ref{methods}. We release the full dataset in \texttt{QAOAKit}~\cite{datadirectlink1,datadirectlink2}.

The median parameters to be transferred are provided in Table~\ref{table:median}.
This one fixed set of median parameters is then transferred to each weighted graph following Eqs.~\ref{eq:scale_weighted_beta}--\ref{eq:scale_weighted_gamma}. Additionally, we have experimented with using rescaled fixed-angle conjecture parameters from \cite{wurtz2021fixedangle}. Because the median parameters provided better performance, we do not include the results using these fixed-angle conjecture parameters.

The approximation ratios obtained by using transferred and directly numerically optimized parameters are presented in Fig.~\ref{fig:approximation_ratios_transfer}.  We observe a median (over all $p$ and weight distributions) difference in approximation ratios between QAOA with directly optimized and transferred median parameters of only 2 percentage points ({\pp}).
As expected, the difference grows with depth $p$ (1.2 {\pp}~for $p=1$, 2.4 {\pp}~for $p=2$, and 3.0 {\pp}~for $p=3$). We also observe that the effectiveness of the parameter transfer depends significantly on the distribution of weights: uniform on $(0,1)$ is the easiest for parameter transfer (1.0 {\pp}), uniform on $(-1,1)$ being harder (2.1 {\pp}), and exponentially drawn weights leading to the hardest instances (3.6 {\pp}). We observe that transfer works well (the gap is small) if the variance of the weights is low. As the variance increases, the performance of the transfer procedure deteriorates (i.e., the gap increases). 
The relationship between the performance of QAOA with transferred parameters and the variance of the weights is shown in Fig.~\ref{fig:gap_vs_std_weight} for $p=1$ and $|V|=14$ nodes, with a similar behavior observed for other values of $p$ and $|V|$. We note that even in the worst case of Fig.~\ref{fig:approximation_ratios_transfer} ($p=3$, exponentially drawn weights), the decrease in approximation ratio from using median parameters is only 4.5 {\pp}  This demonstrates that the concentration of optimized parameters, previously observed for unweighted MaxCut and the SK model, is present in more general weighted MaxCut instances if they are rescaled appropriately. 
Note that all the statistics at each $p$ are obtained by using a \textit{single} value of median parameters, irrespective of the graph size and the edge weights distribution.

\begin{figure}[h]
    \centering
    \includegraphics[width=0.48\linewidth,trim=0 0 0 0in]{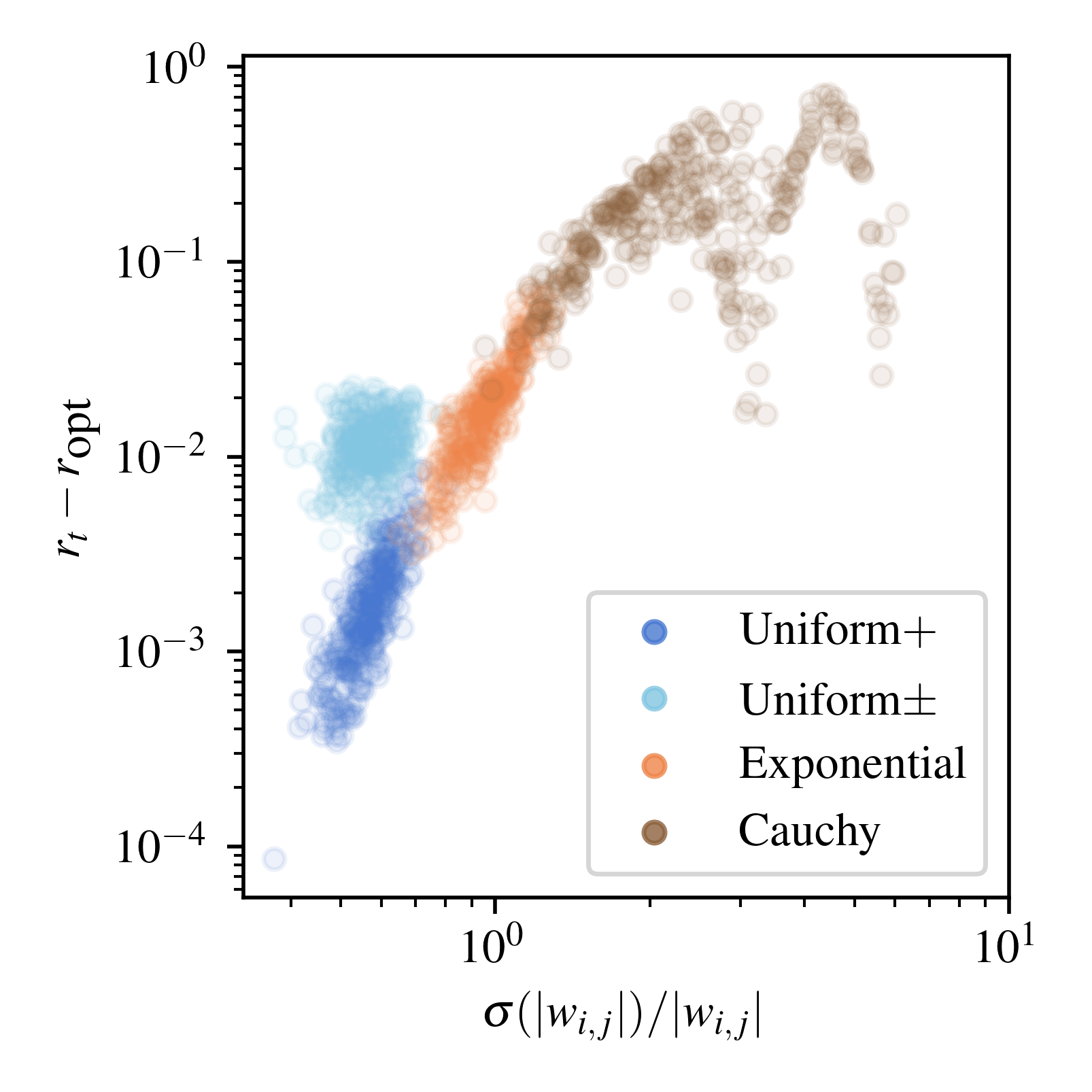}
    \caption{Gap between the QAOA approximation ratio with transferred ($r_t$) and directly optimized ($r_{\text{opt}}$) parameters and the standard deviation of the weights for $p=1$ and 14-node graphs. If the standard deviation is small, parameter transfer works well, and the gap $r_t - r_{\text{opt}}$ is small.
    }
    \label{fig:gap_vs_std_weight}
\end{figure}

We observe that the transferred parameters are sufficiently close to the optimized ones that a single local optimization run nearly always recovers an approximation ratio equivalent to the ratio achieved by the directly optimized parameters. Specifically, for each graph we perform one BFGS optimization run using the transferred median parameters as the starting point, and we observe that the optimized approximation ratio is recovered up to high precision ($10^{-7}$) for $96.35\%$ of instances considered in Fig.~\ref{fig:approximation_ratios_transfer}. Moreover, in $96.34\%$ of the cases, we were able to use parameter symmetries to show that the parameters recovered by the BFGS are the same as the directly  optimized ones up to $10^{-4}$. 
This result indicates that the transferred median parameters are in the same attraction basin as the optimized ones, suggesting that using them as the initial guess can greatly simplify parameter optimization even in the cases where very high-quality QAOA parameters are desired.

The performance of parameter transfer can be understood further by examining how far the transferred parameters are from directly optimized ones. To compute the distances, we first map the optimized parameters to symmetry-related parameters that minimize $||\bm \beta_\mathrm{opt} - \bm \beta_\mathrm{transf}||$ as described in Appendix \ref{methods}. This gives a unique distance for each instance. The Euclidean distances between directly optimized parameters and transferred medians are given in Fig.~\ref{fig:eucledian_distance_transf_and_cauchy} (left).  The distances increase with $p$, partly because of the increase in the dimension of the parameter space in which the distance is computed. As expected from \eqref{eq:scale_weighted_beta}, we consistently observe a small distance between directly optimized and transferred $\bm\beta$. For the parameter $\bm\gamma$, we observe that the distance is larger for distributions on which the performance of the transferred parameters is worse. While the difference varies based on the distributions, the fact that the absolute value of the difference remains small ($<0.1$ median) for all the distributions in Fig.~\ref{fig:approximation_ratios_transfer} explains the success of BFGS in retrieving the optimized parameters.

Intuitively, it is unlikely that one set of median parameters could be transferred to \emph{any} instance by a simple rescaling. To test the limits of parameter transfer, we consider weights sampled from a truncated Cauchy distribution. This generates edge weights with varying orders of magnitude, leading to a significantly higher standard deviation of weights than with the other distributions and also greater variations between randomly drawn weight instances, as seen in Fig.~\ref{fig:gap_vs_std_weight}. The performance of parameter transfer on problems with weights drawn from a Cauchy distribution is given in Fig.~\ref{fig:eucledian_distance_transf_and_cauchy} (right).  Here direct parameter transfer from our scheme is less effective, as anticipated in connection with the objective landscape of Fig.~\ref{fig:contours small}(d). Often, however, the transferred parameters still serve as a good initial guess for local optimization, where they return final approximation ratios equivalent to those from direct optimization in $37.44\%$ of the cases. Hence, the transfer is useful for minimizing the cost of parameter optimization even in this extreme setting. 

We observe much larger distances between directly optimized and transferred $\bm\beta$ for the Cauchy weight distribution. However,  of all large ($>0.1$) distances in $\bm\beta$, $99.9\%$ correspond to the cases where the optimized $\bm\gamma$ has both negative and positive components. This is in contrast to the other cases where typically all parameters have the same sign.

\begin{figure}[t]
    \centering
    \includegraphics[width=0.35\linewidth,trim=0 0 0 0in]{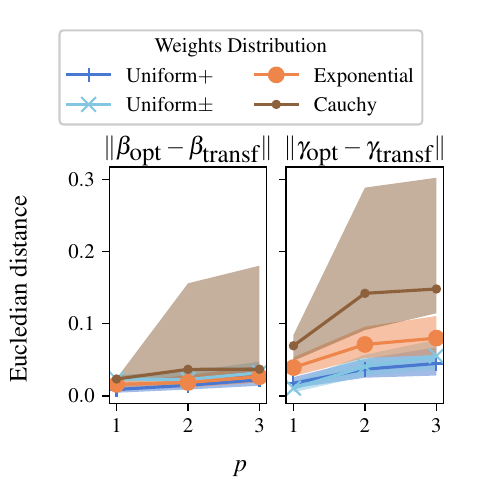}
    \hfil
     \includegraphics[width=0.35\linewidth]{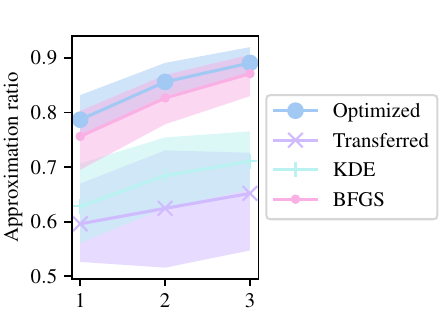}
    \caption{Left: Euclidean distance between transferred and directly optimized parameters, with small distances indicative of parameter concentration. Markers indicate medians, shaded regions are interquartile ranges, and lines guide the eye. Some of the growth with $p$ is due to the growing dimensionality of the parameter space. Right:
    Approximation ratios for the truncated Cauchy weight distribution in each of our approaches. Markers indicate medians,  shaded regions are interquartile ranges, and lines guide the eye. Local optimization with BFGS, starting from the transfer parameters, typically recovers approximation ratios equivalent to full optimization.
    \label{fig:eucledian_distance_transf_and_cauchy}}
\end{figure}

To further improve the performance of parameter transfer, we follow Ref.~\cite{khairy2019learning} and train a generative model that learns a metadistribution of optimized QAOA parameters using kernel density estimation (KDE). We train the model using the same dataset of scaled preoptimized parameters for 9-node graphs given by \eqref{eq:scale_median}, with details in Appendix \ref{methods}.  For each graph, we sample 10 parameters $(\bm\beta, \bm\gamma)$ from the model and use them in addition to the median parameters, for a total of 11 parameters evaluated, and report the best approximation ratio.  At the modest cost of only 10 extra samples we can significantly reduce the gap between the performance with transferred and directly optimized parameters. The median (over the full dataset) gap is reduced by 40\%, from 2.0 {\pp}~to 1.2 {\pp}, with KDE reducing the gap over using only the median parameters in 89.5\% of the cases. The trends observed for median parameters persist when using the KDE-based method: the performance decreases with depth (0.7 {\pp}~for $p=1$, 1.2 {\pp}~for $p=2$ and 1.8 {\pp}~for $p=3$) and with harder distribution (0.6 {\pp}~for uniform on $(0,1)$, 1.0 {\pp}~for uniform on $(-1,1)$, and 2.7 {\pp}~for exponential). 
This corresponds to a median reduction of the gap by 25--50\%, with the worst-case median ($p=3$, exponential distribution) reduced to 3.5 {\pp} 
By construction, the KDE-based technique always gives an approximation ratio better than or equal to using only the median.

\section{Discussion}

Our results demonstrate that
parameter concentration previously established for unweighted MaxCut~\cite{Lotshaw2021BFGS,Galda2021transfer,wurtz2021fixedangle,brandao2018concentration} and the SK model \cite{Farhi2020SK} also applies to instances of weighted MaxCut, provided the problem is appropriately rescaled. At each $p$, the parameters from a single QAOA circuit (obtained by taking the median over a large dataset) can be rescaled and transferred across tens of thousands of instances of weighted MaxCut with a performance that is near to that of QAOA with directly optimized parameters. This significantly reduces the time-consuming and difficult step of finding high-quality parameters in QAOA for weighted MaxCut.

The evidence we present considers depth $p \leq 3$.  Larger $p$ may be necessary for future applications. However, several recent results show that for unweighted MaxCut, similar parameters work well for most instances even at larger depths~\cite{wurtz2021fixedangle,zhou2020quantum}. Based on the successes for the $p$ tested and the small changes we observe as $p$ increases, it is likely that this may be true for weighted MaxCut provided the parameters are rescaled to account for weights. Similarly, previous results for unweighted MaxCut \cite{brandao2018concentration} and the SK model~\cite{2110.10685} have shown that parameter concentration persists as $|V|$ increases, consistent with our observations for weighted MaxCut across the $|V|$ we tested. This suggests that parameter concentration in weighted MaxCut may extend to larger $|V|$ and $p$ that are needed for quantum computational advantage in scientifically relevant problems. 

The extreme example of Cauchy distributed weights highlights the limitations of our parameter transfer approach specifically and ``typical case'' results in general. 
Parameter transfer in a typical case, over a large ensemble of instances, does not necessarily indicate success in  interesting practical cases. 
It remains to be seen whether problems that are relevant to science and industry exhibit properties amenable to parameter transfer and, if not, whether approaches can be developed to overcome these limitations.

\section*{Data availability}

The dataset generated during the current study is available in \texttt{QAOAKit}. The direct URLs to the datafiles are given in Refs.~\cite{datadirectlink1,datadirectlink2}.

\section*{Code availability}

The code that reproduces the figures presented in this manuscript from the published data (direct URL~\cite{codefigsdirectlink}) and the code that performs parameter transfer (direct URL~\cite{codeexampledirectlink}) are available in \texttt{QAOAKit}.

\section*{Acknowledgments}
P.C.L., J.O., and T.S.H.~were supported by the Defense Advanced Research Project Agency ONISQ program under award W911NF-20-2-0051. J.O.~acknowledges the Air Force Office of Scientific Research award, AF-FA9550-19-1-0147 and the National Science Foundation award OMA-1937008. R.S.~and J.L.~were supported by the U.S.\ Department of Energy (DOE), Office of Science, Office of Advanced Scientific Computing Research AIDE-QC and FAR-QC projects. R.S.~was supported by the Argonne LDRD program under contract number DE-AC02-06CH11357.  We gratefully acknowledge the computing resources provided on Bebop, a high-performance computing cluster operated by the Laboratory Computing Resource Center at Argonne National Laboratory.

\section*{Author Contributions}

R.S.~and P.L.~conceived of the project idea. P.L., R.S., and J.L.~developed the simulation code and performed the numerical experiments. R.S.~analyzed the data. R.S., P.L., and J.L.~wrote the manuscript. All authors contributed to technical discussions.

\printbibliography
\appendix
\section{Details of the instances, parameter optimization, and analysis}\label{methods}

In our numerical results we consider weights drawn from differing distributions to probe how QAOA performance depends on the variance of the weights and the presence of negative weights. We use uniform distributions of positive weights, uniform distributions of positive and negative weights, an exponential distribution of positive weights, and a truncated Cauchy distribution with negative and positive weights:
\begin{align} 
p_{\mathrm{uniform}+}(w_{i,j}) \propto 1,\\
p_{\mathrm{uniform}\pm}(w_{i,j}) \propto 1,\\
p_\mathrm{exp}(w_{i,j}) \propto e^{-\alpha w_{i,j}}\\ p_\mathrm{trunc.\ Cauchy}(w_{i,j}) \propto \frac{1}{1+w_{i,j}^2}.
\end{align}
We use rejection sampling to generate samples from the exponential and truncated Cauchy. For the exponential distribution we use a cutoff $w^\mathrm{exp}_{i,j} \leq 16/\alpha$, and without loss of generality we set $\alpha=1$.  In contrast, an exponential distribution on $[0,\infty]$ would have a probability $\int_{16/\alpha}^\infty e^{-\alpha w} dw = e^{-16} \approx 10^{-7}$ to sample a $w_{j,k}^\mathrm{exp} > 16/\alpha$. For the truncated Cauchy distribution we generate samples with $-10^3 \leq w_{i,j} \leq 10^3$.  This gives samples with high variance but is distinct from the full Cauchy distribution on $[-\infty,\infty]$.

For our QAOA parameter optimization runs we use a standard approach based on generating random initial $\bm \beta$ and $\bm \gamma$ and then optimizing $\langle C(\bm \beta, \bm \gamma)\rangle$ with the 
BFGS algorithm to find a local optimum. We expect our results to not be specific to BFGS and to apply to other choices of local optimizers capable of obtaining high-quality parameters $\bm \beta, \bm \gamma$. 

For $|V|=8$ we used the BFGS implementation of Ref.~\cite{NumericalRecipesBFGS}. We compared this with the l-BFGS (low storage BFGS) implementation of NLopt \cite{NLopt,Nocedal1980lBFGS, LiuNocedal1989lBFGS} in test cases   and found that each approach gave satisfactory final results, whereas the approach of Ref.~\cite{NumericalRecipesBFGS} gave final parameters that would more frequently vary significantly from the initial parameters. For $n=14$ and $n=20$, we used the NLopt implementation, since we found this took less compute time for these instances. The initial parameters were sampled uniformly at random from $-\pi/4 \leq \beta_l \leq \pi/4$, which covers the full range of the $\beta_l$ up to symmetries \cite{Lotshaw2021BFGS,zhou2020quantum}, and with $-\pi/\wbar \leq \gamma_l \leq \pi/\wbar$.  This focuses on the region around the maximum at small $\gamma_l$ for all $l$ because larger ranges tended to produce inferior results. Note that since the small values of $\gamma_l$ do not always maximize $\langle C(\bm \gamma, \bm \beta)\rangle$, the optimized parameters are not guaranteed to be globally optimal. For all distributions except Cauchy, for each graph we optimized from 50, 200, and 1,500 random initial parameters at $p=1,2,3$, respectively; for the truncated Cauchy distribution we used 200, 500, and 3,000 seeds, respectively.  The final optimized results for each graph were set to be the parameter set with the highest approximation ratio. 

To learn the metadistribution of optimized QAOA parameters, we use kernel density estimation. We train the model using optimized parameters  $\{\vect{x}_j = (\bm\beta^{s}_j, \bm\gamma^{s}_j)\}_{j=1}^N$ for unweighted MaxCut on all $261,080$ nonisomorphic 9-node graphs~\cite{Lotshaw2021BFGS,shaydulin2021qaoakit}, where the parameters have been scaled according to \eqref{eq:scale_median}. We use the Parzen--Rosenblatt smooth density estimate with a Gaussian kernel, with probability density function given by 
\begin{equation} \label{eq:kde_density}
\hat{f}_{\bm{X}}(\vect{x}) = \frac{1}{N} \sum_{i=0}^{N-1} K_\omega (\vect{x},\vect{x}^i_*),
\end{equation}
where $K_\omega (\vect{x},\vect{x}^i_*) = \frac{1}{(2 \pi \omega^2)^{N/2}} \text{exp}(\frac{-(\vect{x}-\vect{x}^i_*)^T(\vect{x}-\vect{x}^i_*)}{2\omega^2})$. The kernel bandwidth $\omega$ is optimized by using a grid search to maximize the log-likelihood of the out-of-sample data using 5-fold cross-validation. To sample from the density function \eqref{eq:kde_density}, we sample a uniformly random point $\vect{x}_j$ and add to it a sample drawn from a multivariate Gaussian distribution with zero mean and diagonal covariance matrix $\omega^2 I$ with the appropriate identity matrix $I$. We use the scikit-learn~\cite{scikit-learn} implementation of KDE. We release the code and the resulting trained KDE model in \texttt{QAOAKit}~\cite{codeexampledirectlink}.

\subsection*{Proof $\bm\gamma$ scales with objective value}

\begin{theorem}\label{thm:gamma_scaling}
Consider QAOA applied to two objective functions $\mathcal{C}_w$ and $\mathcal{C}$, $\mathcal{C}_w = w\mathcal{C}$ with $w > 0$. Let $r$ be the QAOA approximation ratio for $\mathcal{C}$ with parameters $(\bm\beta,\bm\gamma)$. %
Then QAOA with parameters $(\bm\beta,\bm\gamma/w)$ achieves an identical approximation ratio for $\mathcal{C}_w$, i.e. $r_w=r$. 
\end{theorem}
\begin{proof}

Let $C$ and $C_w$ be the diagonal Hamiltonians representing the objective functions $\mathcal{C}$ and $\mathcal{C}_w$, respectively. From the definition $C_w=wC$, 
\begin{align}
\ket{\bm \beta, \bm \gamma/w, \mathcal{C}_w} & = \prod_{l=1}^p e^{-i \beta_l B} e^{-i \frac{\gamma_l}{w} C_w} \vert+ \rangle = \prod_{l=1}^p e^{-i \beta_l B} e^{-i \gamma_lC} \vert+ \rangle  = \ket{\bm \beta, \bm \gamma, \mathcal{C}}.
\end{align}
The objectives for the two states are related by
\begin{align}
    \langle C_w(\bm \beta, \bm \gamma/w)\rangle & = \bra{\bm \beta, \bm \gamma/w, \mathcal{C}_w} C_w \ket{\bm \beta, \bm \gamma/w, \mathcal{C}_w} \nonumber \\
    & = w \bra{\bm \beta, \bm \gamma/w, \mathcal{C}_w} C \ket{\bm \beta, \bm \gamma/w, \mathcal{C}_w} = w \bra{\bm \beta, \bm \gamma, \mathcal{C}} C \ket{\bm \beta, \bm \gamma, \mathcal{C}} =  w\langle C(\bm \beta, \bm \gamma)\rangle 
\end{align}
The maximum and minimum objective values for $\mathcal{C}$ and $\mathcal{C}_w$ are related by the constant factor $w$, $w\mathcal{C}_{\max} = (\mathcal{C}_w)_{\max}$ and $w\mathcal{C}_{\min} = (\mathcal{C}_w)_{\min}$.  Thus the approximation ratios achieved by QAOA applied to $\mathcal{C}$ and $\mathcal{C}_w$ with parameters $(\bm\beta, \bm\gamma)$ and $(\bm\beta, \bm\gamma/w)$ respectively are identical:
\begin{align}
    r & = \frac{\bra{\bm \beta, \bm \gamma, \mathcal{C}} C \ket{\bm \beta, \bm \gamma, \mathcal{C}} - \mathcal{C}_{\min}}{\mathcal{C}_{\max}- \mathcal{C}_{\min}} = \frac{w}{w} \frac{\bra{\bm \beta, \bm \gamma/w, \mathcal{C}_w} C_w \ket{\bm \beta, \bm \gamma/w, \mathcal{C}_w} - (\mathcal{C}_w)_{\min}}{(\mathcal{C}_w)_{\max}-(\mathcal{C}_w)_{\min}} = r_w.
\end{align}
\end{proof}

\subsection{Parameter symmetries and the parameter distance}

QAOA parameters are known to exhibit a variety of symmetries \cite{zhou2020quantum,Lotshaw2021BFGS} that give identical objective values for symmetry-related parameters.  The first symmetry that is pertinent here is time-reversal symmetry, $(\bm \beta, \bm \gamma) \to (-\bm \beta, -\bm \gamma)$.  This is observed, for  example, in Fig.~\ref{fig:contours small} as a reflection about the origin. The second symmetry is periodicity in $\bm \beta$,  $(\bm \beta, \bm \gamma) \to (\bm \beta \pm \pi/2, \bm \gamma)$.

To define a meaningful angle distance and obtain our results in Fig.~\ref{fig:eucledian_distance_transf_and_cauchy} (left), we must account for these symmetries.  The reason is that two identical sets of parameters, up to symmetries, may have a large distance if they are in different symmetry sectors.  The distance is then an arbitrary factor, related to symmetry sector that is chosen in each instance rather than to characteristic differences in the parameters themselves. To avoid ambiguities in the distance that arise from the freedom to choose from among symmetry-related parameters, for each optimized $(\bm \beta_\mathrm{opt}, \bm \gamma_\mathrm{opt})$, we use the symmetries above to define a set of parameters $(\bm \beta'_\mathrm{opt}, \bm \gamma'_\mathrm{opt})$ that minimizes the distance $||\bm\beta_\mathrm{opt}' - \bm\beta_\mathrm{trans}||$.  This is done by considering all combinations of permutations of the component parameters $(\beta_\mathrm{opt})_l \to (\beta_\mathrm{opt})_l \pm \pi/2$ and time reversal symmetry on the joint $(\bm \beta_\mathrm{opt}, \bm \gamma_\mathrm{opt})$. The result is a unique distance that minimizes $||\bm\beta_\mathrm{opt}' - \bm\beta_\mathrm{trans}||$ from among all choices of symmetry-related parameters.

\section*{Disclaimer}

This paper was prepared for information purposes with contributions from the Future Lab for Applied Research and Engineering (FLARE) Group of JPMorgan Chase \& Co. and its affiliates, and is not a product of the Research Department of JPMorgan Chase \& Co. JPMorgan Chase \& Co. makes no explicit or implied representation and warranty, and accepts no liability, for the completeness, accuracy or reliability of information, or the legal, compliance, tax or accounting effects of matters contained herein. This document is not intended as investment research or investment advice, or a recommendation, offer or solicitation for the purchase or sale of any security, financial instrument, financial product or service, or to be used in any way for evaluating the merits of participating in any transaction.

\vfill
\center{
\framebox{\parbox{.85\linewidth}{
\footnotesize
The submitted manuscript has been created by UChicago Argonne, LLC, Operator of
Argonne National Laboratory (``Argonne''). Argonne, a U.S.\ Department of
Energy Office of Science laboratory, is operated under Contract No.\
DE-AC02-06CH11357. The manuscript is also authored by UT-Battelle, LLC under Contract No. DE-AC05-00OR22725 with the U.S. Department of Energy.
The U.S.\ Government retains for itself, and others acting on its behalf, a
paid-up nonexclusive, irrevocable worldwide license in said article to
reproduce, prepare derivative works, distribute copies to the public, and
perform publicly and display publicly, by or on behalf of the Government.  The
Department of Energy will provide public access to these results of federally
sponsored research in accordance with the DOE Public Access Plan.
http://energy.gov/downloads/doe-public-access-plan.}}}

\end{document}